\documentclass[10pt, conference, letterpaper]{IEEEtran}

\ifCLASSOPTIONcompsoc
  \usepackage[nocompress]{cite}
\else
  \usepackage{cite}
\fi
\usepackage{cite,url}
\usepackage{booktabs}
\usepackage{amsmath,amssymb,bm,dsfont,bbm,mathtools}
\usepackage{algorithmic}
\usepackage{amsthm,nccmath,lipsum,cuted}
\usepackage[dvips]{color}
\usepackage{graphicx}
\usepackage[lined,boxed,commentsnumbered, ruled]{algorithm2e}
\usepackage{subcaption}
\usepackage{cleveref}
\usepackage{amsfonts}
\usepackage{ragged2e}
\newtheorem{theorem}{Theorem}
\newtheorem{lemma}{Lemma}
\newtheorem{definition}{Definition}

\newtheorem{corollary}{Corollary}

\usepackage{graphicx}
\usepackage{verbatim}
\usepackage[footnote,draft,danish,silent,nomargin]{fixme}

\def\P{{\mathbb P}}  
\def\E{{\mathbb E}}  

\def\TD{D}

\def\xm{x_\text{min}}
\def\tmin{\theta_\text{min}}
\def\tm{\theta_\text{max}}

\DeclareMathOperator*{\argmin}{arg\,min}

\graphicspath{{./figs/}}

\makeatletter
\newcommand{\pushright}[1]{\ifmeasuring@#1\else\omit\hfill$\displaystyle#1$\fi\ignorespaces}
\newcommand{\mathleft}{\@fleqntrue\@mathmargin0pt}
\newcommand{\mathcenter}{\@fleqnfalse}
\makeatother
\begin{document}

\title{On the Minimum Achievable Age of Information for General Service-Time Distributions}
 \author{
Jaya Prakash Champati, Ramana R. Avula, Tobias J. Oechtering, and James Gross\\
Information Science and Engineering, KTH Royal Institute of Technology, Stockholm, Sweden\\
	E-mail:$\{$jpra, avula, oech, jamesgr$\}$@kth.se
}
\makeatletter
\def\ps@IEEEtitlepagestyle{
	\def\@oddfoot{\mycopyrightnotice}
	\def\@evenfoot{}
}
\def\mycopyrightnotice{
	{\footnotesize
		\begin{minipage}{\textwidth}
			\centering
			Copyright~\copyright~2020 IEEE. This work is accepted to be published in IEEE INFOCOM 2020. Personal use of this material is permitted. However, permission to use this   
			material for any other purposes must be obtained from the IEEE by sending a request to pubs-permissions@ieee.org.
		\end{minipage}
	}
}
\maketitle

\begin{abstract}
There is a growing interest in analysing the freshness of data in networked systems.
Age of Information (AoI) has emerged as a popular metric to quantify this freshness at a given destination. There has been a significant research effort in optimizing this metric in communication and networking systems under different settings. In contrast to previous works, we are interested in a fundamental question, what is the minimum achievable AoI in any single-server-single-source queuing system for a given service-time distribution? To address this question, we study a problem of optimizing AoI under service preemptions. Our main result is on the characterization of the minimum achievable average peak AoI (PAoI). We obtain this result by showing that a fixed-threshold policy is optimal in the set of all randomized-threshold causal policies. We use the characterization to provide necessary and sufficient conditions for the service-time distributions under which preemptions are beneficial.
\end{abstract}

\section{Introduction}\label{sec:intro}
Future networked systems are expected to provide information updates in real time to support the emerging time-critical applications in cyber-physical systems, the increasing demand for live updates by mobile applications, etc.
Since \textit{freshness} of the information updates is crucial to the performance of the applications, one has to account for it in the design of the networked systems. Age of Information (AoI), proposed in~\cite{kaul_2011a}, has emerged as a relevant performance metric for quantifying the freshness of the updates from the perspective of a destination. It is defined as the time elapsed since the generation of freshest update available at the destination. Unlike the system delay, AoI accounts for the frequency of generation of updates by a source, since it linearly increases with time until an update with latest generation time is received at the destination. Whenever such an update is received AoI resets to the system delay of that update and thus indicating its age. 


Given the above properties and its relevance to the networked systems, the question of how to optimize AoI in a given system has received significant attention in the recent past. The problem of computing optimal arrival rate to minimize some function of AoI has been studied for a given inter-arrival time and service time distribution, e.g., see~\cite{kaul_2012b,yates_2012a,Bacinoglu_2015a,Champati_2018a,Huang_2015}. While the objective function was the average AoI   in~\cite{kaul_2012b,yates_2012a,Bacinoglu_2015a}, the authors in~\cite{Champati_2018a} considered the AoI violation probability, and the authors in~\cite{Huang_2015} considered the average peak AoI (PAoI). Given the sequence of arrivals, the authors in~\cite{Bedewy_2017a} proved that a preemptive last-generated-first-served policy results in smaller age processes at all nodes of a network when the service times are exponential. 


In contrast to the above works, we consider the generate-at-will source model, studied in~\cite{yates_2015a,Sun_2017}, in a single-source-single-server system. Under this model, the source can generate an update at any time instant specified by a \textit{scheduling policy} and thus the arrival sequence here is a function of the policy. Further, under this model no queueing is required, because by the defintion of AoI, at any time instant, sending an old update from a queue would be suboptimal to sending a freshly generated update. A counter-intuitive result is that the work-conserving zero-wait policy, that generates a packet immediately when the server becomes idle, is not optimal for minimizing the average AoI~\cite{yates_2015a,Sun_2017}. In fact, introducing \textit{waiting time} after an update is served was shown to have a lower average AoI. Given a service-time distribution with finite mean and assuming no service preemptions, the authors in~\cite{yates_2015a} solved for optimal-waiting times for minimizing the average AoI, while the authors in~\cite{Sun_2017} solved the problem for any non-decreasing function of AoI. Motivated by the fact that allowing service preemptions could further reduce AoI in this system, we ask a fundamental question \textit{what is the minimum achievable AoI in a single-source-single-server queuing system for any given service-time distribution?} 

In this work, we answer this question for minimum achievable average PAoI\footnote{{\color{black} Minimum achievable average AoI was recently studied in~\cite{Arafa_2019} and is an open problem.}} by considering service preemptions, where the service of an update is preempted and dropped whenever a new update is generated by the scheduling policy. The service times across updates are independent and identically distributed (i.i.d.) with a general distribution (possibly with infinite mean\footnote{In fact, preemptions are more beneficial when the service-time distribution has infinite mean.}). Average PAoI was first studied in~\cite{Costa_2016} for M/M/1/1 and M/M/1/2* systems, and has received considerable attention in recent works~\cite{Huang_2015,Qing2016,Chao2019}, which use non-preemptive service model. The related work on service preemptions is discussed and contrasted with our results in Section~\ref{sec:related}.   

We note that a decision about when to generate a new update that preempts an update under service clearly depends on the service-time distribution and could potentially depend on the past decisions. Thus, minimizing the average PAoI under preemptions results in an infinite-horizon average cost Markov Decision Problem (MDP) where the state space and the action space are continuous. In general, for such a problem, it is hard to prove the existence of an optimal stationary deterministic policy among all randomized causal policies that use the entire history of available information~\cite{krishnamurthy2016partially}. Our key result is that, a work-conserving \textit{fixed-threshold policy}, that chooses a fixed duration for preemptions, minimizes the average PAoI among all randomized-threshold causal policies. 

We prove the above result in two steps. First, we formulate an MDP with appropriate cost functions and show that the policy for choosing the sequence of thresholds between any two AoI peaks is independent of the initial state and is also stationary. Second, we define costs for each decision within the two AoI peaks and show that the sequence of decisions converge to a stationary policy and that a fixed-threshold policy achieves the minimum cost. 
Given the optimal policy among randomized-threshold causal policies, we characterize the minimum average PAoI in any single-source-single-server queuing system.
We also present a necessary and sufficient condition for service-time distributions under which preemptions are always beneficial. Finally,  using a case study we provide an insight for the design of the threshold. 

The rest of the paper is organized as follows. In Section~\ref{sec:prob} we formulate the average PAoI minimization problem. In Section~\ref{sec:preliminary} we present preliminary results that are used in Section~\ref{sec:opt} to obtain the optimal fixed-threshold policy. In Section~\ref{sec:preempt} we discuss the conditions under which preemptions are beneficial. The related work on service preemptions is presented in Section~\ref{sec:related}. In Section~\ref{sec:numerical} we present some numerical results and finally conclude in Section~\ref{sec:conclusion}.

\section{System Model and Problem Statement}\label{sec:prob}
We study an information retrieval system shown in Figure~\ref{fig:system}, where a monitor (e.g., a mobile application) strives to obtain latest information (e.g., newsfeeds) from a source which evolves independently. The source instantaneously generates an information update (or simply update) and sends it to the preemptive server  whenever it receives a request from the monitor. We assume zero delay for a request from the monitor to the source. 
However, an update incurs a random service time, denoted by $X$, at the server before it reaches the monitor. We assume that the service times across the updates are i.i.d. Further, we consider that a new update always preempts an update under service. Note that the above model also holds for a system where the monitor just indicates to the source if an update  was received (for instance by an ACK), and then the source decides itself about when to generate the next update.
Let $F_X(\cdot)$, $f_X(\cdot)$ and $\mathbb{E}[X]$ denote the cumulative distribution function, probability density function and the mean of $X$, respectively. We use $x_\text{min} \geq 0$ to denote the minimum value in the support of $X$. 

\begin{figure}[t]
	\centering
	\includegraphics[scale=.45]{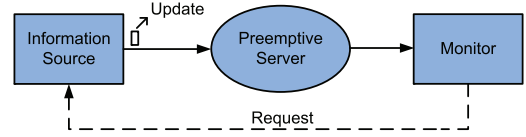}
	\vspace{-.2cm}
	\caption{A model for information retrieval with independently evolving source.}
	\label{fig:system}
	\vspace{-.4cm}
\end{figure}

Let $n$ denote the index of a request and its corresponding update. At any time, the monitor aims to have the freshest update. Note that this depends on the time instants at which monitor requests new information. A scheduling policy for information requests specifies these time instants. To be precise, a scheduling policy $\bm{s} \triangleq \{S_n,n \geq 1\}$, where $S_n \in \mathbb{R}_{\geq 0}$ denotes the generation time of request $n$ (and thus $S_n$ also represents the generation time of update $n$). Using the convention that request $1$ is sent at time zero, the waiting time between requests $n$ and $n+1$, denoted by $Z_n$, is given by $Z_n = S_{n+1} - S_{n}$. Note that the scheduling policy can be  equivalently written as $\bm{s} = \{Z_n,n \geq 1\}$. In the following we describe the policies of interest.
\begin{itemize}
	\item \textit{Work-conserving policy:} $Z_n = \min(\theta_n,X_n)$, for all $n$, where $\theta_n$
	is a threshold for preemption and takes values from $[\xm,\infty) \cup \{\infty\}$. Under this policy a request is sent immediately after an update is received and thus no server idle time is allowed.
	\item \textit{Threshold policy:} $Z_n = \min(\theta_n,X_n)$, for all $n$, where $\theta_n \in [\tmin,\tm]$
	is a threshold for preemption, $\tmin > \xm$ and $\tm < \infty$. A threshold policy is a work-conserving policy with finite thresholds.
	\item \textit{Fixed-threshold policy:} $Z_n = \min(\theta,X_n)$, for all $n$, for some $\theta \in [\tmin,\tm]$. We use $\bm{s}_\theta$ to denote this policy.
	\item \textit{$\xm$-threshold policy:} $Z_n = \xm$, for all $n$. We use $\underline{\bm{s}}$ to denote this policy.
	\item \textit{Zero-wait policy:} $Z_n = X_n$, for all $n$. We use $\bm{s}_\text{Z}$ to denote this policy. Under $\bm{s}_\text{Z}$ a request is sent immediately after an update is received and no preemptions are allowed. 
	We note that $\bm{s}_\text{Z}$ is the only non-preemptive work-conserving policy, where $\theta_n = \infty$, for all $n$. 
\end{itemize}


Let $D_n$ denote the time at which information update $n$ is received at the monitor. We assign $D_n = \infty$, if the update $n$ is dropped due to preemption. We have
\[
D_n = \begin{cases} 
S_n + X_n & \text{if update $n$ is received} \\
\infty & \text{otherwise} 
\end{cases}
\]
In this system, the AoI at the monitor at any time $t$, denoted by $\Delta(t)$, is given by
\begin{align}\label{eq:AoI}
\Delta(t) = t - \underset{n\in\mathbb{N}}{\max}\{S_n: D_n \leq t\}.
\end{align}
Here, $\Delta(t)$ increases linearly with $t$ and drops instantaneously when an update is received. Let $k$ denote the $k$th AoI peak, and $A_k(\bm{s})$ denote the corresponding PAoI value. 
Further, let $n_k$ denote the index of the update received just after the $k$th AoI peak. Note that between updates $n_k$ and $n_{k+1}$ there could be multiple updates that are preempted.
We now have $A_k(\bm{s}) = \Delta(D^{-}_{n_k})$,
where $D^{-}_{n_k}$ is the time just before update $n_k$ is received under $\bm{s}$. We illustrate the above defined quantities in Figure~\ref{fig:AoI_samplepath}, where we present a sample path of AoI under service preemptions. Here, we have used the convention that, a packet is received at time zero and the initial AoI $\Delta(0) = X_0$.
\begin{figure}[t]
	\centering
	\includegraphics[width = 3.5in, height = 1.7in]{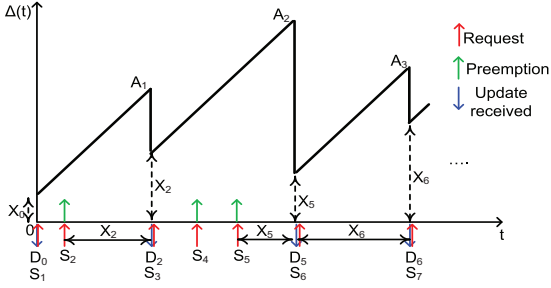}
	\vspace{-.7cm}
	\caption{A sample path of AoI under service preemptions.}
	\label{fig:AoI_samplepath}
	\vspace{-.4cm}
\end{figure}

Under a given policy $\bm{s}$, the average PAoI is defined as
\begin{align}
\zeta(\bm{s}) \triangleq \lim\limits_{K\rightarrow\infty} \frac{1}{K}\E_{\bm{s}}\bigg[\sum\limits_{k=1}^{K} {A}_{k}(\bm{s})\bigg], \label{eq:PAOI_definition}
\end{align}
where the expectation above is taken with respect to a probability distribution determined by $\bm{s}$ and the distribution of $X$. Let $\mathcal{S}$ denote the set of all admissible causal policies for which the limit in (\ref{eq:PAOI_definition}) exists. We are interested in solving the PAoI minimization problem 
\begin{equation*}
\begin{aligned}
& \mathcal{P} \coloneqq \underset{\bm{s} \in \mathcal{S}}{\text{minimize}} \quad \zeta(\bm{s}),
\end{aligned}
\end{equation*}
We use $\bm{s}^*$ to denote an optimal policy, and $\zeta^*$ to denote the minimum average PAoI.

\section{Threshold Policies and Auxiliary Results}\label{sec:preliminary}
In this section we define different classes of threshold policies and provide some important auxiliary results which will be used in the later parts of the paper. In the following, $I_n$ denotes the causal information available at $n$th request.
\begin{definition}
	A \textit{randomized-threshold causal policy} specifies a probability distribution for choosing $\theta_n \in [\tmin,\tm]$ using $I_{n}$ which might be different at each $n$.
\end{definition} 
Let $\mathcal{S}_\text{T}$ denote the set of all randomized-threshold causal policies. The constraint $\theta_n \in [\tmin,\tm]$ is an artefact introduced to bound the MDP costs and facilitate the proof of convergence of the optimal policy to a stationary fixed-threshold policy. However, considering $\xm<\tmin$\footnote{An optimal policy $\bm{s}^*$ never chooses a $\theta_n \! < \!\xm$. Thus, the constraint $\xm<\tmin$ only excludes the case $\theta_n\!=\!\xm$.}
and $\tm<\infty$ excludes $\xm$-threshold policy and zero-wait policy from $\mathcal{S}_{\text{T}}$. 
Nevertheless, for a given problem, choosing $\tmin$ arbitrarily close to $\xm$ and $\tm$ sufficiently large, the imposed constraints result in only a mild restriction of $\mathcal{S}_\text{T}$. This is illustrated in Figure~\ref{fig:policyspace}.

\begin{definition}
	A \textit{repetitive randomized-threshold policy} is a randomized-threshold causal policy under which the joint distributions for choosing the set of thresholds between any two AoI peaks are identical.
\end{definition}
Let $\mathcal{S}_\text{TR}$ denote the set of all repetitive randomized-threshold policies, 
$\mathcal{S}_\theta$ denote the set of all fixed-threshold policies. From the above definitions, we have $\mathcal{S}_\theta \subset \mathcal{S}_\text{TR} \subset \mathcal{S}_\text{T} \subset \mathcal{S}$.

{\allowdisplaybreaks From Figure~\ref{fig:AoI_samplepath}, it is easy to infer that under any policy $\bm{s}$, we have, for all $k$,
	\begin{align}\label{eq:Apeak}
	A_{k+1}(\bm{s}) &= \TD_{n_{k+1}} - S_{n_k} \nonumber\\
	&= \underbrace{\TD_{n_{k+1}} - \TD_{n_k}}_{ \triangleq \, Y_{k+1}(\bm{s})} + \underbrace{\TD_{n_k} - S_{n_k}}_{ \triangleq \, \check{X}_k(\bm{s})}.
	\end{align}
}
Note that $\check{X}_k(\bm{s})$ is equal to $X_{n_k}$, the service time of update $n_k$. However, under preemptive policies $\check{X}_k(\bm{s})$ does not have the same distribution as $X$. The time $Y_{k+1}(\bm{s})$ denotes the duration between the time instances at which update $n_{k}$ and $n_{k+1}$ are received.
Note that $Y_{k+1}(\bm{s})$ constitutes the idle time of the server after reception of update $n_k$.  
Therefore, introducing idle time penalizes PAoI and it is always beneficial to send a request immediately after receiving an update. This implies that an optimal policy belongs to the set of work-conserving policies. 
Hence, we arrive at the following lemma.
\begin{figure}[t]
	\centering
	\includegraphics[width = 3.5in, height = 1.7in]{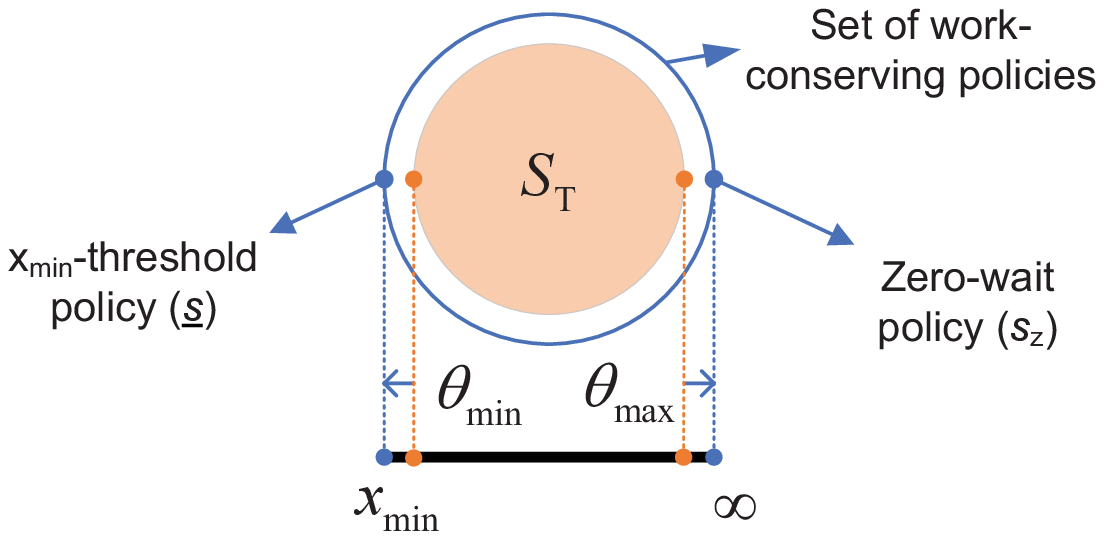}
	\vspace{-.7cm}
	\caption{Visualization of $\mathcal{S}_\text{T}$ under the constraint $\tmin\leq\theta_n\leq\tm$, where $\tmin>\xm$ and $\tm<\infty$.}
	\label{fig:policyspace}
	\vspace{-.4cm}
\end{figure}
\begin{lemma}\label{lem:threshold_policies}
	The optimal policy $\bm{s}^*$ belongs to the set of work-conserving policies.
\end{lemma}
In the following, we present some auxiliary results that will be extensively used in the proofs later in Section~\ref{sec:opt}. We first define deterministic-repetitive threshold policies and compute $\zeta(\bm{s})$ for this calss of policies.
\begin{definition}
	A \textit{deterministic-repetitive-threshold policy} uses the same sequence of deterministic thresholds between any two AoI peaks.
\end{definition}
Let $\{\theta_i, i \geq 1\}$ denote a sequence of deterministic thresholds. Then, a deterministic-repetitive-threshold policy $\bm{s}$ repeats this sequence between any two peaks. In the following lemma we characterize $\check{X}_k(\bm{s})$ and $Y_{k+1}(\bm{s})$. 

\begin{lemma}\label{thm:zetaExpr}
	For a deterministic-repetitive-threshold policy $\bm{s}$, $\check{X}_k(\bm{s})$ are i.i.d. with mean $\mathbb{E}[\check{X}(\bm{s})]$, and  $Y_{k+1}(\bm{s})$ are i.i.d. with mean $\mathbb{E}[Y(\bm{s})]$, where 
	\begin{align}\label{eq:checkX}
	\mathbb{E}[\check{X}(\bm{s})] =& \!\int_0^{\theta_1}\!\!\!\! x f_X(x) dx + 
	\sum_{j=1}^{\infty}\!\prod_{i=1}^j \P\{X_{i}\!\! >\! \theta_{i}\}\!\! \int_0^{\theta_{j+1}}  \!\!\!\! x f_X(x) dx,
	\end{align}
	\begin{align}\label{eq:Y}
	\mathbb{E}[Y(\bm{s})]\! =& \mathbb{E}[\check{X}(\bm{s})]\! + 
	\!\sum_{j=1}^{\infty}\prod_{i=1}^j \! \P\{X_{i}\!\! >\! \theta_{i}\}F_X(\theta_{j+1}) \!\!\sum_{i=1}^{j} \theta_{i}, 
	\end{align}
	and $\zeta(\bm{s}) = \mathbb{E}[\check{X}(\bm{s})] + \mathbb{E}[Y(\bm{s})].$
\end{lemma}
\begin{proof}
	The proof is given in Appendix~\ref{proof:thm:zetaExpr}.
\end{proof}
Using the result in Lemma~\ref{thm:zetaExpr} we compute $\zeta(\bm{s}_\theta)$, the average PAoI under a fixed-threshold policy.
\begin{corollary}\label{cor1}
	For a fixed-threshold policy $\bm{s}_\theta$, we have the average PAoI $\zeta(\bm{s}_\theta) = \mathbb{E}[\check{X}(\bm{s}_\theta)] + \mathbb{E}[Y(\bm{s}_\theta)]$, where
	\begin{align}
	\mathbb{E}[\check{X}(\bm{s}_\theta)] &= \frac{\int_0^{\theta} xf_X(x) dx}{F_X(\theta)} \label{eq:zeta_1},\\
	\mathbb{E}[Y(\bm{s}_\theta)] &=  \frac{\theta - \int_{0}^{\theta} F_X(x) dx}{F_X(\theta)}= \mathbb{E}[\check{X}(\bm{s}_\theta)] + \frac{\theta \P(X>\theta)}{F_X(\theta)} \label{eq:zeta_2}.
	\end{align}
\end{corollary}
\begin{proof}
	The proof is given in Appendix~\ref{proof:cor1}.
\end{proof}

\begin{corollary}\label{cor_results}
	For a given distribution $F_X(\cdot
	)$, the average PAoIs achieved by the $\xm$-threshold policy $\underline{\bm{s}}$ and the zero-wait policy ${\bm{s}_{\text{Z}}}$ are given by
	\begin{gather}
	\zeta(\underline{\bm{s}}) = \zeta(\bm{s}_{\xm}), \quad \zeta(\bm{s}_{\text{Z}}) = 2\mathbb{E}[X].
	\end{gather}
\end{corollary}

\section{Minimum Achievable Average PAoI}\label{sec:opt}
In this section we first present a fixed-threshold policy that is optimal among all causal randomized policies. Next, in any single-source-single-server queuing system,  we present the optimal policy among all work-conserving policies and provide an expression for the minimum average PAoI.
\begin{theorem}\label{thm:optimalamongcausal}
	Given the distribution of service times $F_X(\cdot)$, there exists a fixed-threshold policy $\bm{s}_{\theta^\dagger}$ in $\mathcal{S}_\theta$ that is optimal in $\mathcal{S}_{\text{T}}$,
	where $\theta^\dagger $ is the optimal fixed threshold, given by 
	\begin{align}\label{eq:opt_theta}
	\theta^\dagger \triangleq \argmin_{\theta \in[\tmin,\tm]} \; \zeta(\bm{s}_\theta).
	\end{align}
\end{theorem}
\begin{proof}
	The proof of the theorem is given in two steps. First, we formulate an infinite horizon average cost MDP problem equivalent to $\mathcal{P}$ in the domain of $\mathcal{S}_{\text{T}}$ and show that an optimal policy $\bm{s}^\dagger$  belongs to $\mathcal{S}_\text{TR}$. Next, we consider the decision process between two successive updates and show the independence of the optimal policy with the past decisions. Further, we prove that a fixed-threshold $\theta^\dagger$ minimizes the average PAoI. The details are provided in Appendix~\ref{proof:thm:optimalamongcausal}.
\end{proof}

Consider a single-source-single-server queuing system with a given service time distribution, having any arrival process and any service policy, e.g., FCFS/LCFS, preemptions/no preemptions, packet drops/no drops etc. By the definition of AoI, it is easy to argue that the minimum average PAoI in this system will be at least the minimum average PAoI in our system with generate-at-will source model, no queueing, and service preemptions. Now, as illustrated in Figure~\ref{fig:policyspace}, 
for a given problem, by choosing $\tmin$ arbitrarily close to $\xm$ and $\tm$ sufficiently large, the set $\mathcal{S}_{\text{T}}\cup\{\bm{s}_\text{Z}, \underline{\bm{s}}\}$ can closely approximate the set of work-conserving policies. Therefore, from Theorem~\ref{thm:optimalamongcausal} and Lemma~\ref{lem:threshold_policies}, it immediately follows that $\min(\zeta(\bm{s}_{\theta^\dagger}), \zeta(\bm{s}_{\text{Z}}), \zeta(\underline{\bm{s}}))$ is the minimum achievable PAoI. 
Now using Corollay~\ref{cor_results}, we arrive at the following result on minimum achievability. 
\begin{theorem}\label{thm:LB}
	In any single-source-single-server queuing system with i.i.d. service times, and a given distribution $F_X(\cdot)$, the minimum achievable average PAoI is given by
	\begin{gather}
	\zeta^* = \min(\zeta(\bm{s}_{\theta^\dagger}), 2\mathbb{E}[X], \zeta(\bm{s}_{\xm})),
	\end{gather}
	and thus, the optimal policy $\bm{s}^*$ is either $\bm{s}_{\theta^\dagger}$ or $\bm{s}_{\text{Z}}$ or $\underline{\bm{s}}$, whichever achieves $\zeta^*$. 
\end{theorem}

\section{When are Preemptions Beneficial?}\label{sec:preempt}
In this section we study the conditions under which preemptions are beneficial, i.e., allowing preemptions will result in a stricly lower average PAoI. 
From Theorem~\ref{thm:LB}, a necessary and sufficient condition for preemptions to be beneficial is as follows:
\begin{align}\label{eq:cond}
\exists \, \theta \geq 0 \text{ such that } \min(\zeta(\bm{s}_{\theta^\dagger}),\zeta(\bm{s}_{\xm})) < 2\mathbb{E}[X].
\end{align}
In the following we consider an example distribution and obtain the condition under which preemptions are beneficial.

\textbf{\textit{Case Study:}} Consider a random service time $X$ that takes value $t_1$ with probability $p$ and $t_2$ with probability $1-p$, and $0 < t_1 < t_2$. Note that, here $\xm = t_1$ and threrefore $\zeta(\bm{s}_{\xm}) = t_1(1+p)/p$.
The distribution of $X$ can be written as follows:
\begin{align*}
f(x) &= p\delta(x-t_1) + (1-p)\delta(x-t_2),\\
F_X(x) &= pu(x-t_1) + (1-p)u(x-t_2),
\end{align*}
where $\delta(\cdot)$ and $u(\cdot)$ are Dirac delta function and unit-step function, respectively. Note that for this distribution choosing threshold $\theta < t_1$ or $\theta > t_2$ does not reduce average PAoI. Therefore, we compute $\zeta(\bm{s}_\theta)$ for $t_1 < \theta \leq t_2$.
\begin{align*}
\zeta(\bm{s}_\theta) &= \frac{\int_{t_1}^\theta xf(x) dx}{F_X(\theta)} + \frac{\theta - \int_{t_1}^\theta F_X(x)dx}{F_X(\theta)}\\
&= \frac{t_1 p}{p} + \frac{\theta - p(\theta - t_1)}{p}\\
&= \frac{2pt_1 + (1-p)\theta}{p} > t_1(1+p)/p\; \text{ for all } \theta > t_1.
\end{align*}
From the last step above we conclude that $\min(\zeta(\bm{s}_{\xm}),\zeta(\bm{s}_{\theta^\dagger}))=\zeta(\bm{s}_{\xm})$. This implies that, under preemptive policies 
whenever an update is not received within the duration $t_1$, it is optimal to send a new request just after $t_1$.

We use~\eqref{eq:cond} to check if preemptions are beneficial or not. Since $\mathbb{E}[X] = p t_1 + (1-p)t_2$, preemptions are beneficial iff $\zeta(\bm{s}_{\xm}) < 2\mathbb{E}[X]$, which implies
\begin{align}\label{eq:tau2}
t_2 > \frac{t_1}{1-p}\left[1+\frac{1}{p}-2p\right].
\end{align}
The condition in~\eqref{eq:tau2} establishes a lower bound on $t_2$ for preemptions to be beneficial. For example, if $p = \frac{1}{2}$ and $t_1 = 1$, then preemptions are beneficial if $t_2$ is greater than $2$. 


Note that the service-time distribution in the above example is simple enough to compute $\theta^\dagger$ analytically and use~\eqref{eq:cond} to infer whether preemptions will be beneficial or not. In general, it is not straightforward to do so for any service-time distribution. In the following lemma we provide a sufficient condition that could be used to infer if preemptions are beneficial for a given class of distributions.
\begin{lemma}\label{lem:cond}
	For any single-source-single-server queueing system, a sufficient condition for preemptions to be beneficial for minimizing average PAoI is as follows: 
	\begin{align*}
	\exists \, \theta \geq 0 \text{ such that } \mathbb{E}[X] < \mathbb{E}[X-\theta|X>\theta] + \frac{\theta}{2}.
	\end{align*}
\end{lemma}
\begin{proof}
	From~\eqref{eq:cond}, a sufficient condition is that there exists $\theta$ such that
	\begin{align*}
	&\zeta(\bm{s}_{\theta}) < 2\mathbb{E}[X] \nonumber\\
	\overset{(a)}{\Leftrightarrow} & 2\mathbb{E}[\check{X}(\bm{s}_\theta)] + \frac{\theta \P(X > \theta)}{F_X(\theta)} < 2\mathbb{E}[X] \nonumber\\
	\overset{(b)}{\Leftrightarrow} & 2 \mathbb{E}[X]\! + \! 2 \!\!\int_0^{\theta}\!\!\!\!\! xf_X(x)dx\! +\! \theta \P(X \!>\! \theta)\! <\! 2F_X(\theta)\mathbb{E}[X]\! +\! 2 \mathbb{E}[X]\nonumber\\
	\Leftrightarrow & 2\P(X > \theta) \mathbb{E}[X] +  \theta \P(X > \theta) < 2 \int_{\theta}^{\infty} xf_X(x)dx\nonumber\\
	\overset{(c)}{\Leftrightarrow} & \mathbb{E}[X] +  \frac{\theta}{2}  < \frac{\int_{\theta}^{\infty} xf_X(x)dx}{\P(X > \theta)}\nonumber\\
	\Leftrightarrow & \mathbb{E}[X] < \mathbb{E}[X-\theta|X>\theta] + \frac{\theta}{2}.
	\end{align*}
	In step $(a)$ we have used $\zeta(\bm{s}_{\theta}) = \mathbb{E}[\check{X}(\bm{s}_\theta)] + \mathbb{E}[Y(\bm{s}_\theta)]$ and~\eqref{eq:relation_zeta}. In step $(b)$ we have added $2\mathbb{E}[X]$ on both sides. We arrive at the final step by using the following equation in step $(c)$.
	\begin{align*}
	\mathbb{E}[X-\theta|X>\theta]\! = \frac{\int_{\theta}^{\infty} (x-\theta) f_X(x) dx}{\P(X > \theta)}\! = \frac{\int_{\theta}^{\infty}\!\! xf_X(x) dx}{\P(X > \theta)}\! -\! \theta.
	\end{align*}
\end{proof}
From Lemma~\ref{lem:cond}, we infer that a sufficient condition is the existence of a $\tau$ that satisfies $\mathbb{E}[X-\tau|X>\tau] > \mathbb{E}[X]$. This condition implies that given an elapsed time $\tau$, the expected residual should be greater than the mean value. This is satisfied by heavy-tailed distributions and hyper-exponential distributions~\cite{Moorsel2006}. 

\section{Related Work}\label{sec:related}
Most of the works in the AoI literature that considered service preemptions focused on analysing the average AoI and average PAoI for different queueing systems, e.g., see~\cite{kaul_2012a,Chen2016,Soysal2019,Najm_2016,Najm_2017,Najm_2018}.
In contrast, the authors in~\cite{Veeraruna2018} studied the problem of whether to preempt or not preempt the current update in service in an M/GI/1/1 system with the objective of minimizing the average AoI. They established conditions under which two extreme policies \textit{always-preempt} and \textit{no-preemptions} are optimal among stationary randomized policies. 

The work by the authors in~\cite{Arafa_2019} is contemporary to ours. They studied the same system model as ours but considered the problem of minimizing the average AoI in the system.
In the following we first summarise their results and then contrast our contributions with theirs. Considering a fixed-threshold policy for doing preemptions, the authors  first solve for an optimal \textit{waiting time}\footnote{The idle time of the server after an update is received. Idling the server does not reduce the average PAoI but may reduce the average AoI.}. Stating that it is hard to obtain a closed-form expression for the average AoI in terms of the fixed threshold and its corresponding optimal waiting time, the authors compute, numerically, the optimal fixed threshold for two service-time distributions, namely, exponential and shifted exponential. It was not shown that the proposed method would result in a global optimum solution for general service-time distribution. In our work, we considered the average PAoI minimization problem. We have derived a fixed-threshold policy $\bm{s}_{\theta^\dagger}$ that is optimal in the set of randomized causal policies. This result provides a justification for the choice of fixed-threshold policies in~\cite{Arafa_2019}. Furthermore, using $\bm{s}_{\theta^\dagger}$, zero-wait and $\xm$-threshold policies we have characterized the minimum achievable average PAoI.

In their seminal work~\cite{Luby1993}, the authors studied the problem of finding optimal thresholds for restarting the execution of an algorithm having random runtime. For discrete service-time distributions the authors provided an optimal fixed-threshold policy that minimizes the expected run-time, considering the set of stationary randomized policies. 
Compared to the problem in~\cite{Luby1993}, minimizing expected PAoI is hard as the consecutive AoI peaks are not independent even under a stationary policy. Furthermore, we have proven a general result since we considered the set of randomized causal policies and continuous service-time distributions.

\section{Numerical Analysis}\label{sec:numerical}
In this section, we compute the optimal fixed threshold for the Erlang and Pareto service-time distributions. We have considered the Pareto distribution to illustrate the effectivenes of preemptions for heavy-tailed distributions, and the Erlang distribution is chosen due to the fact that it models a tandem of exponential (memoryless) servers. We compare the average peak AoI achieved by zero-wait policy, optimal fixed-threshold policy $\bm{s}_{\theta^\dagger}$, and median-threshold policy that uses the median as the fixed threshold. We study the median-threshold policy because it can be useful in cases where the distribution of the service times is not known apriori but the median can be estimated. Further,  unlike mean, median is always finite and is an unbiased estimate. 

\subsection{Erlang Service-Time Distribution}
Erlang distribution is characterized by two parameters $\{k,\lambda\}$, where $k$ is the shape parameter and $\lambda$ is the rate parameter. In Figure~\ref{fig:Erlang_xm1_varying_theta}, we plot the average PAoI $\zeta(\bm{s}_\theta)$, computed using Corollary~\ref{cor1}, by varying the threshold $\theta$. The minimum values of $\zeta(\bm{s}_\theta)$ are indicated by the points in magenta. Recall that, for $k=1$ the Erlang distribution results in an exponential distribution. For this case, from Figure~\ref{fig:Erlang_xm1_varying_theta} we observe that the function $\zeta(\bm{s}_\theta)$ is concave, and therefore the optimal $\theta^\dagger$ approaches zero which further implies that $\bm{s}^*$ always chooses the threshold zero. In contrast, for $k \geq 2$, the functions are convex in $\theta$ and we obtain $\bm{s}^* = \bm{s}_{\theta^\dagger}$. We have observed this change in the nature of $\zeta(\bm{s}_\theta)$ with different parameter values of a distribution in the case of log-normal, but it is not presented here due to space limitation. In Figure~\ref{fig:Erlang_xm1_comparison}, we compare the average peak AoI achieved by different policies. We observe that in general zero-wait policy has average PAoI close to $\zeta(\bm{s}^*)$. This is because the sufficient condition that $\mathbb{E}[X-\theta|X > \theta] > \mathbb{E}[X]$ is not satisfied by the Erlang distribution for any $\theta$~\cite{Moorsel2006}, and thus allowing preemptions does not significantly reduce average PAoI. The average PAoI under median-threshold policy is relatively higher and also diverges from both zero-wait and $\bm{s}^*$ when $k$ increases, thus suggesting that using preemptions with arbitrary threshold could in fact penalize the average PAoI. Thus, it is important to verify first if preemptions are beneficial for a given service-time distribution. The conditions provided in~\eqref{eq:cond} and Lemma~\ref{lem:cond} are potentially useful toward this end. 
\begin{figure}[t]
	\centering
	\includegraphics[width = 2.5in]{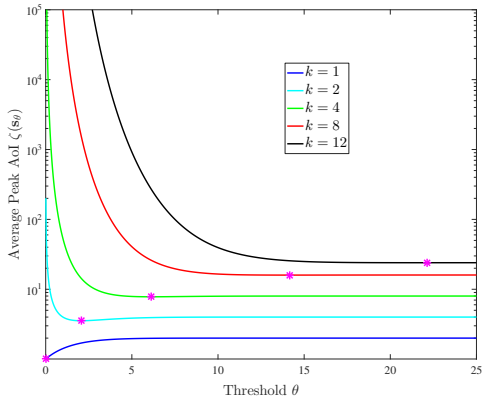}
	\vspace{.1cm}
	\caption{Average peak AoI vs. $\theta$ under the Erlang service-time distribution for different $k$ and  $\lambda = 1$.}
	\label{fig:Erlang_xm1_varying_theta}
\end{figure}

\begin{figure}[t]
	\centering
	\includegraphics[width = 2.5in]{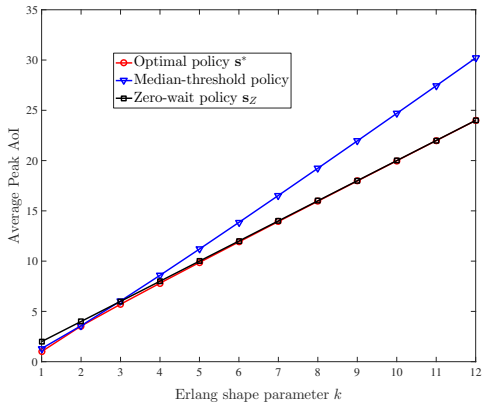}
	\vspace{.1cm}
	\caption{Average peak AoI achieved by different policies under the Erlang service-time distribution with varying $k$ and $\lambda = 1$.}
	\vspace{-.4cm}
	\label{fig:Erlang_xm1_comparison}
\end{figure}

\subsection{Pareto Service-Time Distribution}
The Pareto distribution is characterized by two parameters $\{x_m,\alpha\}$, where $x_m$ is the scale parameter and $\alpha$ is the tail index. The smaller the $\alpha$, the heavier the tail. 
In Figure~\ref{fig:Pareto_xm1_varying_theta}, we plot the average PAoI by varying the threshold $\theta$. The minimum values of $\zeta(\bm{s}_\theta)$ are indicated by the points in magenta. Observe that in this case $\zeta(\bm{s}_\theta)$ are convex in $\theta$ for each $\alpha$. Further, for the Pareto distribution we obtain $\bm{s}^* = \bm{s}_{\theta^\dagger}$. In Figure~\ref{fig:Pareto_xm1_comparison}, we compare the average peak AoI achieved by different policies. Observe that for higher $\alpha$ values the optimal policy coincides with zero-wait policy because the distribution has a light tail. For $\alpha \leq 1$, the distribution has a heavy tail and infinite mean, and thus zero-wait policy also attains this value. In contrast, the optimal policy achieves finite average PAoI values in this case, and this illustrates the effectiveness of preemptions for heavy-tailed distributions. Furthermore, the median-threshold policy performs consistently well when compared with the optimal policy and thus it is an attractive choice when the parameters $\{x_m,\alpha\}$ are not known apriori, but an estimate of the median is available. 
\begin{figure}[t]
	\centering
	\includegraphics[width = 2.5in]{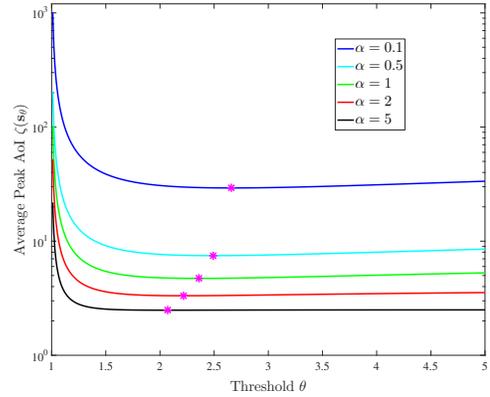}
	\vspace{.2cm}
	\caption{Average peak AoI vs. $\theta$ under the Pareto service-time distribution for different $\alpha$ and $x_m = 1$.}
	\label{fig:Pareto_xm1_varying_theta}
\end{figure}

\begin{figure}[t]
	\centering
	\includegraphics[width = 2.5in]{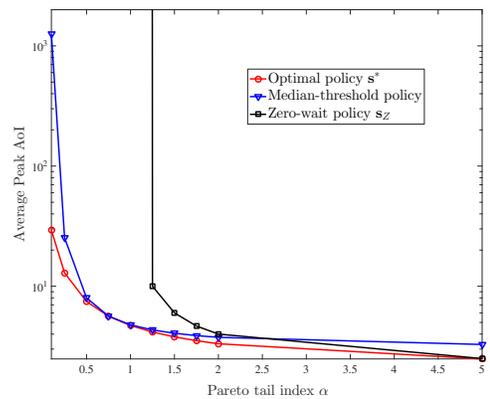}
	\vspace{.2cm}
	\caption{Average peak AoI achieved by different policies under the Pareto service-time distribution with varying $\alpha$ and $x_m\! =\! 1$.}
	\vspace{-.4cm}
	\label{fig:Pareto_xm1_comparison}
\end{figure}

\section{Conclusion}\label{sec:conclusion}
In this work we have studied a problem of finding the minimum achievable average PAoI for a given service-time distribution. To this end, we have considered generate-at-will source model and service preemptions. Using an MDP formulation we have shown that a fixed-threshold policy achieves minimum average PAoI in the set of randomized-threshold causal policies. The minimum achievable average PAoI in any single-source-single-server queuing system is then given by the minimum average PAoI achieved among zero-wait, $\xm$-threshold and the optimal fixed-threshold policies. Using the fact that zero-wait policy is optimal among all non-preemptive policies, we establish necessary and sufficient conditions for the service-time distributions under which preemptions result in a lower average PAoI. In the numerical analysis, we have used the Pareto service-time distribution to illustrate the effectiveness of preemptions for heavy-tailed distributions.

We leave the numerical analysis studying the average PAoI for wide range of service-time distributions for future work. We plan to study the minimum achievability for other functions of AoI including the average AoI.

\appendix
\subsection{Proof of Lemma~\ref{thm:zetaExpr}}\label{proof:thm:zetaExpr}
We first analyse $\check{X}_k(\bm{s})$ and $Y_{k+1}(\bm{s})$. Recall that $n_{k}$ is the index of the $k$th received update. We note that at time $D_{n_{k-1}}$, request $n_{k-1}+1$ will be sent and update $n_{k-1}+1$ will be generated by the source and sent to the server. Note that $\bm{s}$ repeats the same sequence $\{\theta_i, i \geq 1\}$ between any two peaks. If $X_{n_{k-1}+1} \leq \theta_1$ then update $n_{k-1}+1$ will be received successfully. In this case, we set $n_k = n_{k-1}+1$ and $\check{X}_k(\bm{s})=X_{n_{k-1}+1}$. If $X_{n_{k-1}+1} > \theta_1$, then update $n_{k-1}+1$ will be preempted by sending request $n_{k-1}+2$. In this case the above statements can be similarly repeated by comparing $X_{n_{k-1}+2}$ and $\theta_2$. Using the above analysis we characterize $\check{X}_k(\bm{s})$ in terms of the service times of updates $\{n_{k-1}+1,n_{k-1}+2,\ldots\}$, and the corresponding thresholds $\{\theta_{1},\theta_{2},\ldots\}$. 
\[
\check{X}_k(\bm{s}) = 
\begin{cases} 
X_{n_{k-1}+1} & X_{n_{k-1}+1} \leq \theta_1  \\
X_{n_{k-1}+2} & X_{n_{k-1}+1} > \theta_1, X_{n_{k-1}+2} \leq \theta_2  \\
X_{n_{k-1}+3} & X_{n_{k-1}+1} > \theta_1, X_{n_{k-1}+2} > \theta_2, \\
&X_{n_{k-1}+3} \leq \theta_3\\  
\vdots

\end{cases}
\]
Note that the above characterization of $\check{X}_k(\bm{s})$ is true for any $k$ as $\bm{s}$ is a deterministic-repetitive threshold policy. Since $X_n$ are i.i.d. we infer that $\check{X}_k(\bm{s})$ are also i.i.d. In the following we write $\check{X}_k(\bm{s})$ using indicator functions.
\begin{align}\label{eq:Xk+1(s)}
& \check{X}_k(\bm{s}) = X_{n_{k-1}+1}\mathbbm{1}\{X_{n_{k-1}+1} \leq \theta_{1}\} + \nonumber\\
& \sum_{j=1}^{\infty}\prod_{i=1}^j \! \mathbbm{1}\{X_{n_{k-1}+i}\!\! >\! \theta_{i}\}X_{n_{k-1} + j+1}\mathbbm{1}\{X_{n_{k-1} + j+1}\! \leq\! \theta_{j+1}\}
\end{align}
Taking expectation on both sides and noting that $X_{n_{k-1}}+i$ and $X_i$ are i.i.d. we arrive at~\eqref{eq:checkX}. 

To analyse $Y_{k+1}(\bm{s})$, we start with request $n_{k}+1$ that is sent at time $D_{n_{k}}$ and compare its service time $X_{n_{k}+1}$ with $\theta_1$. We use similar analysis as above and characterize $Y_{k+1}(\bm{s})$ as follows.
\begin{equation}\label{eq:Yk+1(s)}
Y_{k+1}(\bm{s}) = 
\left\{ \begin{array}{ll}
X_{n_{k}+1} & X_{n_{k}+1} \leq \theta_1  \\
\theta_1 + X_{n_{k}+2} & X_{n_{k}+1} > \theta_1, X_{n_{k}+2} \leq \theta_2  \\
\theta_1 + \theta_2 + X_{n_{k}+3} & X_{n_{k}+1} > \theta_1, X_{n_{k}+2} > \theta_2, \\
&X_{n_{k}+3} \leq \theta_3\\  
\vdots

\end{array} \right.
\end{equation}

\begin{align*}
& Y_{k+1}(\bm{s}) = X_{n_k +1}\mathbbm{1}\{X_{n_k + 1} \leq \theta_{1}\} + \nonumber\\
& \quad \quad \sum_{j=1}^{\infty}\prod_{i=1}^j \! \mathbbm{1}\{X_{n_k + i}\!\! >\! \theta_{i}\}\Big[ \mathbbm{1}\{X_{n_k + j+1}\! \leq\! \theta_{j+1}\} \sum_{i=1}^{j} \theta_{i} + \nonumber \\
& \quad \quad \quad \quad \quad \quad \quad \quad \quad \quad  \, X_{n_k + j+1}\mathbbm{1}\{X_{n_k + j+1}\! \leq\! \theta_{j+1}\}\Big]\\
&= \check{X}_{k+1}(\bm{s}) +  \sum_{j=1}^{\infty}\prod_{i=1}^j \! \mathbbm{1}\{X_{n_k + i}\!\! >\! \theta_{i}\}\mathbbm{1}\{X_{n_k + j+1}\! \leq\! \theta_{j+1}\} \sum_{i=1}^{j} \theta_{i}
\end{align*}
Again, taking expectation on both sides and noting that $X_{n_{k}}+i$ and $X_i$ are i.i.d. we arrive at~\eqref{eq:Y}. Further, as $\bm{s}$ is a deterministic-repetitive threshold policy and $X_n$ are i.i.d., we infer that $Y_{k}(\bm{s})$ are i.i.d.

Since $\check{X}_k(\bm{s})$ are i.i.d., and $Y_{k}(\bm{s})$ are i.i.d., and $A_{k+1}(\bm{s}) = \check{X}_k(\bm{s}) + Y_{k+1}(\bm{s})$, we conclude that $A_k(\bm{s})$ for all $k$ have identical distribution with mean $\mathbb{E}[\check{X}(\bm{s})] + \mathbb{E}[Y_{k+1}(\bm{s})]$. Therefore,
\begin{align*}
\zeta(\bm{s}) = \lim\limits_{K\rightarrow\infty} \frac{1}{K}\E_{\bm{s}}\bigg[\sum\limits_{k=1}^{K} {A}_{k}(\bm{s})\bigg] 
= \mathbb{E}[\check{X}(\bm{s})] + \mathbb{E}[Y_{k+1}(\bm{s})].
\end{align*} 

\subsection{Proof of Corollary \ref{cor1}}\label{proof:cor1}
Substituting $\theta_i = \theta$ for all $i$ in~\eqref{eq:checkX}, we obtain
\begin{align*}
\mathbb{E}[\check{X}(\bm{s}_{\theta})] &\overset{(a)}{=} \int_0^\theta x dF_X(x) + \sum_{j=1}^{\infty} \P(X > \theta)^{j} \int_0^\theta x dF_X(x) \nonumber \\
&\overset{(b)}{=} \int_0^\theta x dF_X(x)\sum_{j=0}^{\infty} \P(X > \theta)^{j} \overset{(c)}{=} \frac{\int_0^\theta x dF_X(x)}{F_X(\theta)}.
\end{align*}
In step $(a)$ we have used $\mathbb{E}[X\mathbbm{1}\{x \leq \theta\}] = \int_0^\theta x dF_X(x)$. In step $(c)$ we have used the sum for infinite geometric series. 

Similarly, substituting $\theta_i = \theta$ for all $i$ in~\eqref{eq:Y}, we obtain
{\allowdisplaybreaks
	\begin{align}\label{eq:ExpRestart}
	\mathbb{E}[Y(\bm{s}_\theta)] =& \mathbb{E}[\check{X}(\bm{s}_{\theta})] + \sum_{j=1}^{\infty} \P(X > \theta)^{j}F_X(\theta)j\theta \nonumber \\
	\overset{(a)}{=}&  \frac{\int_0^\theta\!\! x dF_X(x)}{F_X(\theta)} \!+\! \theta F_X(\theta)\P(X\!>\!\theta)\!\sum_{j=1}^{\infty} j\P(X \!>\! \theta)^{j-1} \nonumber \\
	\overset{(b)}{=}& \frac{\theta F_X(\theta) - \int_0^\theta F_X(x)dx}{F_X(\theta)}  + \frac{\theta F_X(\theta)\P(X>\theta)}{F_X(\theta)^2} \nonumber \\
	= & \frac{\theta - \int_{0}^{\theta} F_X(x) dx}{F_X(\theta)}.
	\end{align}}

From steps $(a)$ and $(b)$ of~\eqref{eq:ExpRestart} we infer that
\begin{align}\label{eq:relation_zeta}
\mathbb{E}[\check{X}(\bm{s}_\theta)] + \frac{\theta \P(X>\theta)}{F_X(\theta)} = \mathbb{E}[Y(\bm{s}_\theta)]
\end{align}

\subsection{Proof of Theorem~\ref{thm:optimalamongcausal}}\label{proof:thm:optimalamongcausal}
In this proof, we use the notation $F_1^N$ to denote the sequence $[F_1,\dots,F_N]$ and $\mathcal{A}^N$ to denote the N-fold Cartesian product of a set $\mathcal{A}$. Let $I_{k,r} = \{A_1^{k-1},\check{X}_1^{k-1},\tilde{I}_1^{k-1},\theta_{k,1},\dots,\theta_{k,{r-1}}\}$ denote the causal information available to the scheduler at $r$th request after $(k-1)$th update, where $\tilde{I}_{k} = \{\theta_{k,1},\dots,\theta_{k,\check{R}_k}\}$ denotes the sequence of threshold values between $(k-1)$th and $k$th updates and $\check{R}_{k} = n_k - n_{k-1}$. Here, $I_{k,0}$ denotes the information state exactly at $(k-1)$th update. Further, we use $i_{k,r}$ to denote a realization of $I_{k,r}$ and $\delta_{k,r}(i_{k,r})$ to denote the conditional distribution function of the threshold $\theta_{k,r}$ given $i_{k,r}$. Recall that a randomized-threshold causal policy $\bm{s}$ specifies a sequence of causal sub-policies at each update, denoted by $\mu_k(i_{k,0})$, where each $\mu_k$ specifies the conditional distributions $\delta_{k,r}(i_{k,r})$ at each request $r$ between the $(k-1)$th and $k$th updates. For a given $i_{k,0}$, the sub-policy $\mu_k$ belongs to $\mathcal{U}$, which is the set of randomized sub-policies that specify the distributions of thresholds between two successive updates. For a given $i_{k,r}$, the distribution $\delta_{k,r}$ belongs to $\mathcal{F}$, which is the set of valid probability distribution functions.

Now, we solve $\mathcal{P}$ among $\mathcal{S}_{\text{T}}$ in two steps. First, we formulate an infinite-horizon average cost MDP problem with the decision epochs as the times at which the updates are received. In the next step, we consider the decision epochs as the times at which requests are sent between any two successive updates. 
\subsubsection*{\textbf{Step 1}}
The identified infinite-horizon average cost MDP problem equivalent to $\mathcal{P}$ has the following elements:
\begin{itemize}
	\item \emph{State:} the service time of an update, $\check{X}_{k-1} \in \mathbb{R}_{+}$, 
	\item \emph{Action:} the sequence of conditional distribution functions, 
	\begin{gather*}
	\mu_k(i_{k,0})=\big\{\delta_{k,r}(i_{k,r})\big\vert r\in\mathbb{N}\big\}
	\end{gather*}
	\item \emph{Cost function:} the expected PAoI given $i_{k,0}$,
	\begin{flalign*}
	c_{k}(i_{k,0},\mu_k) &= \E_{\mu_k}\big[{{A}_{k}\vert I_{k,0} = i_{k,0}}\big]\\
	&= \check{x}_{k-1} + \E_{\mu_k}\big[{B_{k}+\check{X}_{k}\big\vert I_{k,0} = i_{k,0}}\big],	
	\end{flalign*}
	where ${B}_{k}$ denotes the time lost due to preemptions. 
\end{itemize}
Here, using the result from the Lemma \ref{thm:zetaExpr}, we obtain
\mathleft
\begin{flalign*}
\alpha_X(\mu_k) &=: \E_{\mu_k}\big[{\check{X}_{k}\vert I_{k,0} = i_{k,0}}\big] \nonumber \\
&= \E_{\mu_k}\Bigg[\!\sum_{r=1}^{\infty}\prod_{m=1}^{r-1}\bar{F}_X(\theta_{k,m})\int_{0}^{\theta_{k,r}}\!\!\!xf_X(x)dx\Bigg],
\end{flalign*}
\begin{flalign*}
\beta_X(\mu_k) &=: \E_{\mu_k}\big[{B_{k}\vert I_{k,0} = i_{k,0}}\big] \nonumber \\
&= \E_{\mu_k}\Big[{Y_{k}\vert I_{k,0} = i_{k,0}}\Big] - \E_{\mu_k}\Big[{\check{X}_{k}\vert I_{k,0} = i_{k,0}}\Big]\nonumber\\
&=\E_{\mu_k}\Bigg[\sum_{r=1}^{\infty}\prod_{m=1}^{r}\bar{F}_X(\theta_{k,m})\theta_{k,r}\Bigg],
\end{flalign*}
where $\alpha_X\!\!:\! \mathcal{U} \!\rightarrow\! \mathbb{R},$ and $\beta_X\!\!:\!\mathcal{U}\! \rightarrow \!\mathbb{R}$ are deterministic functions. Therefore, we can express the cost function as
\mathcenter
\begin{flalign} \label{eq:cost_step1}
c_{k}(\check{x}_{k-1},\mu_k) 	&= \check{x}_{k-1} + \alpha_X(\mu_k) + \beta_X(\mu_k).	
\end{flalign}	
Now, the problem $\mathcal{P}$ in the domain of $\mathcal{S}_\text{T}$ is equivalent to the infinite horizon average cost problem given by 
\begin{gather}
\bm{s}^\dagger = \argmin_{\bm{s} \in \mathcal{S}_\text{T}} \Biggl\{ \lim\limits_{K\rightarrow\infty} \frac{1}{K}\E_{\bm{s}}\bigg[\sum\limits_{k=1}^{K} c_{k}(\check{x}_{k-1},\mu_k)\bigg] \Biggr\},
\end{gather}
where $\bm{s}^\dagger$ is the optimal policy.
Note that for a given policy $\bm{s}\in\mathcal{S}_{\text{T}}\subset\mathcal{S}$, we have $\alpha_X(\mu_k)<\infty$ and $\beta_X(\mu_k)<\infty$ because the limit in (\ref{eq:PAOI_definition}) exists for all $\bm{s}\in\mathcal{S}$. Given $\check{x}_{1}$, let $V_K$ denotes the minimum expected cumulative cost over a finite horizon $k = [1,\cdots,K]$ and the optimal finite-horizon solution 
can be obtained using the backward recursion of the stochastic Bellman's dynamic programming \cite{krishnamurthy2016partially} given by
\mathcenter
\begin{gather*}
V_{k}(i_{k,0}) \!=\! \underset{\mu_k\in\mathcal{U}}{\min}\Big\{\!c_{k}(\check{x}_{k-1},\mu_k) +  \E_{\mu_k}\!\Big[V_{k+1}\big\vert I_{k,0} = i_{k,0}\Big]\!\Big\}, 
\end{gather*}
where the value function $V_{k}$ denotes the optimal expected cumulative cost-to-go from $k$ to $K$. Since there will be no cost after the finite-horizon, we initialize the recursion with $V_{K+1} = 0$. Thus, for $k=K$, we have 
\mathcenter
\begin{flalign*}
{V}_K(i_{K,0}) &= \check{x}_{K-1} + \underbrace{{\underset{\mu_K\in\mathcal{U}}{\min}\Big\{\alpha_X(\mu_K)+\beta_X(\mu_K)\Big\}}}_{\tilde{V}_K}
\end{flalign*}
where $\tilde{V}_K$ is a constant for all $i_{K,0}$. Similarly, for $k=K-1$,
\mathcenter
\begin{gather}
{V}_{K-1}(i_{K-1,0}) = \check{x}_{K-2} + \tilde{V}_{K-1} + \tilde{V}_K,
\label{eq:thisequation}
\end{gather}
where 
\mathcenter
\begin{gather*}
\tilde{V}_{K-1} = {\underset{\mu_{K-1}\in\mathcal{U}}{\min}\Big\{\!2\alpha_X(\mu_{K-1})+\beta_X(\mu_{K-1})\!\Big\}}, \nonumber \\
\mu_{K-1}^\dagger = {{\underset{\mu_{K-1}\in\mathcal{U}}{\text{argmin}}\Big\{\alpha_X(\mu_{K-1})+\beta_X(\mu_{K-1})\Big\}}}.
\end{gather*}
%
Here, $\tilde{V}_{K-1}$ is a constant and the optimal sub-policy $\mu_{K-1}^\dagger$ is independent of $i_{K-1,0}$. Now, for some $k = m$ such that $1< m \leq K-1$, we assume that the optimal sub-policy satisfies $\mu_m^\dagger=\mu_{K-1}^\dagger$ and the value function has the same structure as in (\ref{eq:thisequation}), that is given by 
\mathcenter
\begin{gather*}
V_{m}(i_{m,0}) = \check{x}_{m-1} + \textstyle\sum_{l=m}^{K}\tilde{V}_{l}, 
\end{gather*}
where $\tilde{V}_{m}^{K}$ are some constants. Next, for $k=m-1$, we get
\mathleft
\begin{flalign*}
{V}_{k}(i_{k,0}) &= \underset{\mu_{k}\in\mathcal{U}}{\min}\Bigg\{\check{x}_{k-1} + \alpha_X(\mu_{k})+\beta_X(\mu_{k}) + \nonumber\\[-0.1cm]
&\qquad\qquad\quad\,\,\,\,\E_{\mu_{k}}\Bigg[\check{X}_{k} + \sum_{l=k+1}^{K}\tilde{V}_{l}\vert I_{k,0} = i_{k,0}\Bigg]\Bigg\} \nonumber\\[0.2cm]
&= \check{x}_{k-1} + \underbrace{\underset{\mu_{k}\in\mathcal{U}}{\min}\Big\{2\alpha_X(\mu_{k})+\beta_X(\mu_{k})\Big\}}_{\tilde{V}_{k}}+ \sum_{l=k+1}^{K}\tilde{V}_{l},
\end{flalign*}
where $\tilde{V}_{k}$ is a constant for all $i_{k,0}$ and $\mu_k^\dagger = \mu_{K-1}^\dagger$. Therefore, using backward induction, for all $1\leq k<K$, we have that $\mu_k^\dagger = \mu^\dagger$, where $\mu^\dagger$ is independent of $i_{k,0}$ and is given by
\mathcenter
\begin{gather}\label{eq:thiseqn3}
\mu^\dagger = {{\underset{\mu\in\mathcal{U}}{\text{argmin}}\Big\{2\alpha_X(\mu)+\beta_X(\mu)\Big\}}}.
\end{gather}	
Hence, the optimal policy $\bm{s}^\dagger$ that minimizes $\mathcal{P}$ among $\mathcal{S}_{\text{T}}$ specifies $\mu^\dagger$ at each update, independent of the current information, i.e., $\bm{s}^\dagger\in\mathcal{S}_{\text{TR}}$. Thus, the minimum expected PAoI is given by
\mathcenter
{\allowdisplaybreaks
	\begin{flalign}\label{eq:repoptimalpolicy}
	\zeta^\dagger \!=\! \lim\limits_{K\rightarrow\infty} \frac{1}{K}\E_{\mu^\dagger}\Bigg[\sum\limits_{k=1}^{K}c_{k}(\check{X}_{k-1},\mu^\dagger)\Bigg]
	\!=\! 2\alpha_X(\mu^\dagger)+\beta_X(\mu^\dagger).
	\end{flalign}}
\subsubsection*{\textbf{Step 2}}
In the following, we drop the index $k$ and ignore the information $I_{k,0}$, as the optimal policy $\bm{s}^\dagger$ is invariant with respect to $k$ and $I_{k,0}$. Here, we solve (\ref{eq:thiseqn3}) by changing the decision epochs of the MDP problem to the times at which requests are sent between any two successive updates. Let $I'_{r}= \{\theta_{1},\dots,\theta_{r-1}\}$ denote the causal information sequence at $r$th request after an update and $c'$ denotes the cost defined as
\mathcenter
\begin{gather}
c'(\theta_r) = 2\int_{0}^{\theta_r}\!\!\!xf_X(x)dx + \theta_r\bar{F}_X(\theta_r).
\end{gather}
such that, for any $\mu\in\mathcal{U}$, we have 
\mathcenter
\begin{gather}
\zeta(\mu) = 2\alpha_X(\mu) + \beta_X(\mu) = \E_{\mu}\Bigg[\sum_{r=1}^{\infty}\prod_{m=1}^{r-1}\bar{F}_X(\theta_{m})c'(\theta_r)\Bigg]. \label{eq:eqn1}
\end{gather}

Let $\omega = \{\theta_i\vert i \in\mathbb{N}\}$ be a realization of $\mu$ for which, we have the sequence $\{J_r\}$ defined by 
\begin{gather}
J_r = \prod_{m=1}^{r-1}\bar{F}_X(\theta_{m})c'(\theta_r).
\end{gather}
Here, for all $r\geq 1$, $\theta_r\in[\tmin,\tm]$, where $\tmin = \xm + \epsilon$, $\epsilon>0$ and $c'(\theta_r)$ is an increasing function of $\theta_r$. That is, there exists some $\bm{C}<\infty$ such that $0 \leq c'(\theta_r)\leq \bm{C}$. Further, we have $0\leq\bar{F}_X(\theta_{r})<1$ for all $r\geq 1$. Therefore, $J_r \rightarrow 0$ as $r\rightarrow \infty$ and consequently, for a sufficiently large $R$, we have
\begin{gather}
\sum_{r = R+1}^{\infty} J_r \approx 0.  \label{eq:thiseqn}
\end{gather}

Let $\zeta_{R}^\dagger$ be the minimum expected cumulative cost over the finite horizon $[1,\cdots,R]$, which is given by
\mathcenter
\begin{gather}\label{eq:finitehorizon2}
\zeta_{R}^\dagger =  \underset{\delta_1^{R}\in\mathcal{F}^{R}}{\min}\Bigg\{\E_{\delta_1^{R}}\Bigg[\sum_{r=1}^{R}\prod_{m=1}^{r-1}\bar{F}_X(\theta_{m})c'(\theta_r)\Bigg]\Bigg\}.
\end{gather}

Similar to Step 1, the optimal solution to (\ref{eq:finitehorizon2}) can be obtained using the backward recursion of the stochastic Bellman's dynamic programming \cite{krishnamurthy2016partially} given by
\mathcenter
\begin{gather*}
\zeta_{r}(i'_r) \!=\! \underset{\delta_r\in\mathcal{F}}{\min}\Bigg\{\!\E_{\delta_r}\!\Bigg[\prod_{m=1}^{r-1}\bar{F}_X(\theta_{m})c'(\theta_r) + \zeta_{r+1}(I'_{r+1})\Bigg]\!\Bigg\}, 
\end{gather*}	
where the value function $\zeta_{r}$ denotes the optimal expected cumulative cost-to-go from $r$ to $R$. As (\ref{eq:thiseqn}) is true for any realization $\omega$ of $\mu$, we have $\zeta_{R+1}\approx0$. Now, for $r=R$, 
\begin{flalign}
\zeta_R(i'_{R})  =\prod_{m=1}^{R-1}\bar{F}_X(\theta_{m})\underbrace{\underset{\delta_R\in\mathcal{F}}{\min}\bigg\{\E_{\delta_r}\Big[c'(\theta_r)\Big]\bigg\}}_{\tilde{\zeta}_R}.\label{eq:thiseqn4}
\end{flalign}
From (\ref{eq:thiseqn4}), it is easy to see that $\tilde{\zeta}_R$ is a constant and the optimal distribution $\delta_R^\dagger$ is independent of $i'_R$. 
Next, for some $l>1$, we assume that the optimal distribution $\delta_{l}^\dagger$ is independent of $i'_{l}$ and the value function has the same structure as in (\ref{eq:thiseqn4}), that is given by 
\mathcenter
\begin{gather*}
\zeta_{l}(i'_{l}) = \prod_{m=1}^{l-1}\bar{F}_X(\theta_{m})\times\tilde{\zeta}_l,
\end{gather*}
for some constant $\tilde{\zeta}_l > 0$. Next, for $r=l-1$, we have
\mathcenter
\begin{gather}\label{eq:iterq}
\zeta_{r}(i_{r}) =  \prod_{m=1}^{r-1}\!\bar{F}_X(\theta_{m})\underbrace{\underset{\delta_r\in\mathcal{F}}{\min}\Big\{\E_{\delta_r}\!\Big[c'(\theta_r) + \tilde{\zeta}_l\bar{F}_X(\theta_{r})\Big]\!\Big\}}_{\tilde{\zeta}_r}, 
\end{gather}
where $\tilde{\zeta}_r$ is a constant for all $i'_{r}$. Therefore, using backward induction, we have that all $\delta_r^\dagger$ are independent of $i'_r$, where $r\in[1,\dots,R]$. As the backward induction is true for any arbitrarily large $R$, it is also true for the optimal sub-policy $\mu^\dagger$. Next, we drop $i'_r$ and rewrite (\ref{eq:iterq}) in terms of $\tilde{\zeta}_r$ as
\mathcenter
\begin{gather}\label{eq:valuefunction2}
\tilde{\zeta}_r = \underset{\delta_r\in\mathcal{F}}{\min}\Big\{\E_{\delta_r}\Big[c'(\theta_r)+\tilde{\zeta}_{r+1}\bar{F}_X(\theta_r)\Big]\!\Big\},
\end{gather}
Now, let $\theta_r^\dagger$ be given by
\begin{gather}
\theta_r^\dagger = \underset{\theta_r\in[\tmin,\tm]}{\text{argmin}}\Big\{c'(\theta_r)+\tilde{\zeta}_{r+1}\bar{F}_X(\theta_r)\!\Big\},\label{eq:valuefunction3}
\end{gather}  
Here, we denote a deterministic distribution with $\bm{1}_{\theta}$ for which $\mathbb{P}(\theta_r \!=\! \theta) \!=\! 1$. From (\ref{eq:valuefunction3}), at each backward iteration, we have that $\delta_r^\dagger = \bm{1}_{\theta_r^\dagger}$ minimizes (\ref{eq:valuefunction2}) since, for any $\delta_r\in\mathcal{F}$, we have
\begin{gather*}
c'(\theta_r^\dagger)+\tilde{\zeta}_{r+1}\bar{F}_X(\theta_r^\dagger) \leq \E_{\delta_r}\Big[c'(\theta)+\tilde{\zeta}_{r+1}\bar{F}_X(\theta)\Big].
\end{gather*}  
Let $T:\mathbb{R}_{\geq 0}\!\rightarrow\mathbb{R}_{\geq 0}$ be the Bellman's operator, given by
\begin{gather*}
T(U) = \underset{\theta\in[\tmin,\tm]}{\min}\Big\{c'(\theta) + U\bar{F}_X(\theta)\Big\}.
\end{gather*}
Using the similar argument as in \cite[Theorem 7.6.2]{krishnamurthy2016partially}, for any $U_1$ and $U_2$ in $\mathbb{R}_{\geq 0}$, we have
\mathcenter
\begin{flalign*}
\Big|T(U_1) - T(U_2)\Big| \leq \Big|U_1 - U_2\Big|\underset{\theta\in[\tmin,\tm]}{\max}\Big\{\bar{F}_X(\theta)\Big\}.
\end{flalign*}
Therefore, the Bellman's operator forms a contraction mapping for all $\theta\in[\tmin,\tm]$. Using Banach's fixed point theorem, for some $\theta^\dagger\in[\tmin,\tm]$, we have that there exists a unique fixed point $\tilde{\zeta}^\dagger$ to the recursive equation (\ref{eq:valuefunction2}). Similar to the case of an infinite horizon discounted cost MDP problem discussed in \cite[Theorem 7.6.2]{krishnamurthy2016partially}, where the conclusion is that a stationary (but state-dependent) policy is optimal for the infinite-horizon, we conclude that using the fixed-threshold $\theta^\dagger$ at all requests minimizes average PAoI, i.e., there exists an $\bm{s}^\dagger\in\mathcal{S}_{\theta}$. Using Corollary \ref{cor1}, we obtain the optimal $\theta^\dagger$, which is given by
\begin{align}\label{eq:opt_theta_final}
\theta^\dagger \triangleq \argmin_{\theta \in[\tmin,\tm]} \; \zeta(\bm{s}_\theta),
\end{align}
Therefore, the minimum expected PAoI among $\mathcal{S}_{\text{T}}$ is given by 
\begin{flalign*}
\zeta(\bm{s}_{\theta^\dagger}) &= \frac{1}{{F}_X(\theta^\dagger)}\times\Bigg[\!2\int_{0}^{\theta^\dagger}xf_X(x)dx + \theta^\dagger\bar{F}_X(\theta^\dagger)\Bigg]. 
\end{flalign*}	

\newpage

\end{document}